\documentclass[11pt,a4paper]{article}

\usepackage{amsmath,amssymb,amsthm}

\usepackage{fullpage}
\usepackage{cite}
\usepackage{textcomp}
\usepackage{mathrsfs}
\usepackage{graphicx}
\usepackage{color}
\usepackage{float}
\DeclareMathOperator{\tr}{tr}

\DeclareMathOperator*{\essinf}{ess\,inf}
\DeclareMathOperator*\sgn{sgn}

\numberwithin{equation}{section}

\theoremstyle{plain}
\newtheorem{theorem}{Theorem}[section]
\newtheorem{proposition}[theorem]{Proposition}
\newtheorem{corollary}[theorem]{Corollary}
\newtheorem{lemma}[theorem]{Lemma}
\theoremstyle{definition}
\newtheorem{definition}[theorem]{Definition}
\newtheorem{assumption}[theorem]{Assumption}

\theoremstyle{remark}
\newtheorem{remark}[theorem]{Remark}

\usepackage{parskip}
\makeatletter
\def\thm@space@setup{%
  \thm@preskip=\parskip \thm@postskip=5pt
}
\makeatother

\begin{document}

\title{Continuous viscosity solutions to linear-quadratic stochastic control problems with singular terminal state constraint\thanks{Financial support by {\sl d-fine GmbH} is gratefully acknowledged. We thank Paulwin Graewe for many discussions and valuable comments. We thank two anonymous referees for valuable comments and suggestions that greatly helped to improve the presentation of the results.}}

\author{Ulrich Horst\footnote{Department of Mathematics, and School of Business and Economics, Humboldt-Universit\"at zu Berlin
         Unter den Linden 6, 10099 Berlin, Germany; email: horst@math.hu-berlin.de} ~and Xiaonyu Xia \footnote{Department of Mathematics, Humboldt-Universit\"at zu Berlin
         Unter den Linden 6, 10099 Berlin, Germany; email: xiaxiaon@math.hu-berlin.de}}

\maketitle

\begin{abstract}
This paper establishes the existence of a unique nonnegative continuous viscosity solution to the HJB equation associated with a linear-quadratic stochastic control problem with singular terminal state constraint and possibly unbounded cost coefficients. The existence result is based on a novel comparison principle for semi-continuous viscosity sub- and supersolutions for PDEs with singular terminal value. Continuity of the viscosity solution is enough to carry out the verification argument. 
\end{abstract}



\section{Introduction}

Let $T \in (0,\infty)$ and let $(\Omega,\mathcal F,(\mathcal F_t)_{t\in[0,T]},\mathbb P)$ that satisfies the usual conditions and carries a Poisson process $N$ and an independent $\tilde d$-dimensional standard Brownian motion $W$. We analyze the linear-quadratic stochastic control problem 
\begin{equation}\label{P1}
	\essinf_{\xi,\mu} E\left[\int_0^T\eta(Y_s) |\xi_s|^2+\theta\gamma(Y_s)|\mu_s|^2+\lambda(Y_s)|X_s^{\xi,\mu}|^2\,ds\right]
\end{equation}
subject to the state dynamics
\begin{equation}\label{P2}
\begin{split}
	dY^{}_t &= b(Y^{}_t) dt + \sigma(Y^{}_t) dW_t, \quad Y^{}_0=y \\
	dX_t^{\xi,\mu} &=- \xi_t\,dt- \mu_t\,dN_t, \quad X^{\xi,\mu}_0=x 
\end{split}
\end{equation} 
and the terminal state constraint 
\begin{equation}\label{P3}
	X^{\xi,\mu}_{T}=0.
\end{equation}
We assume that { $\theta$ is a positive constant, that} the cost coefficients $\eta,\lambda,\gamma$ are continuous and of polynomial growth, that $\eta$ is twice continuously differentiable and that the diffusion coefficients $b,\sigma$ are Lipschitz continuous. We prove the existence of a unique continuous viscosity solution to the resulting HJB equation and give a representation of the optimal control in terms of the viscosity solution. 

Control problems of the form (\ref{P1})-(\ref{P3}) arise in models of optimal portfolio liquidation under market impact when a trader can simultaneously trade in a primary venue and a dark pool. Dark pools are alternative trading venues that allow investors to reduce market impact and hence trading costs by submitting liquidity that is shielded from public view. Trade execution is uncertain, though, as trades will be settled only if matching liquidity becomes available. In such models, $X^{\xi,\mu}$ describes the portfolio process when the traders submits orders at rates $\xi$ to the primary venue for immediate execution and orders of sizes $\mu$ to the dark pool. Dark pools executions are governed by the {Poisson} process $N$ {with rate $\theta$}.  The process $\eta$ describes the instantaneous market impact; it often describes the so-called market depth. The process $\gamma$ describes adverse selection costs associated with dark pool trading while $\lambda$ usually describes market risk, e.g. the volatility of a portfolio holding.    

Starting with the work of Almgren and Chriss \cite{Almgren2001}  portfolio liquidation problems have received considerable attention in the financial mathematics and stochastic control literature in recent years; see \cite{Ankirchner2014, Graewe2018, Graewe2015, Horst2014, Kratz2014, Kratz2015,Kruse2016,Popier2006,Popier2017,Schied2013a} and references therein for details. From a mathematical perspective one of their main characteristics is the singular terminal condition of the value function induced by the terminal state constraint \eqref{P3}. The constraint translates into a singular terminal state constraint on the associated HJB equation and causes significant difficulties in proving the existence and, even more so, the uniqueness of solutions to that equation.
 
Under a continuity and polynomial growth condition on the cost coefficients $\eta,\lambda,\gamma$ it has been shown in \cite{Graewe2018} that the HJB equation admits at most one continuous viscosity solution of polynomial growth. The proof  used a comparison principle for {\sl continuous} viscosity solutions to PDEs with singular terminal value. Since the comparison principle applies only to continuous functions, it can not be used to establish the {\it existence} of a viscosity solution. Instead, it was shown in \cite{Graewe2018} that a (unique) {\sl classical} solution to the HJB equation exists under strong boundedness and regularity assumptions on the model parameters. In this paper we prove a novel comparison principle for {\sl semi-continuous} viscosity solutions for PDEs with singular terminal value from which we deduce the existence of a continuous viscosity solution to our HJB equation using Perron's method. The existence of a {\sl continuous} viscosity solution is enough to carry out the verification arguments and to give a representation of the optimal control in feedback form. 

There are several papers that provide verification arguments without assuming continuity of viscosity solutions. For instance, a utility optimization problem with delays and state constraints has been considered in \cite{Federico2011}. The authors solved in the viscosity sense the associated HJB equation under the assumption that the utility function satisfies the Inada condition, a condition that is not satisfied in our model. In \cite{Belak2017}, the authors studied the general verification result for stochastic impulse control problems, {\sl assuming} that a comparison principle for discontinuous viscosity solutions of the HJB equation holds. This is a very strong hypothesis that can be avoided in our case. The linear-quadratic structure of our control problem allows us to characterize the value in terms of a PDE without jumps, and the verification argument can be given in terms of the associated FBSDE after the existence of the viscosity solution has been established. 

To the best of our knowledge, existence of {\sl continuous} solutions to HJB equations associated with control problems of the form (\ref{P1})-(\ref{P3}) has so far only been established under $L^\infty$ assumptions on the model parameters. The existence of unique continuous viscosity solution was established when $\eta$ is a constant and $\lambda$ is of polynomial growth in \cite{Ankirchner2012a}. Existence and uniqueness of solutions in suitable Sobolev spaces for bounded stochastic cost and diffusion coefficients was proved in \cite{Graewe2015,Horst2016}; classical solutions were considered in \cite{Graewe2018}. 

The restriction to constant market impact terms and/or bounded impact functions and diffusion coefficients is unsatisfactory. In a portfolio liquidation framework, it is natural to choose a two-dimensional driving factor where the first component is a mean-reverting process, e.g. an Ornstein-Uhlenbeck process that describes a liquidity index and the second component is a geometric Brownian motion with zero drift that describes the dynamics of the unaffected stock price process. It is then natural to chose $\eta$ to be a strictly monotone unbounded function of the liquidity index and $\lambda$ to be the square of the geometric Brownian motion so that market risk is measured by the volatility of the portfolio value. Our results apply to such setting.   

The papers \cite{Ankirchner2014, Kruse2016, Popier2017} allow for unbounded coefficients. They characterize the value function as the {\sl minimal solution} to some BSDE with singular terminal value. BSDEs with singular terminal value were first studied in \cite{Popier2006}. In \cite{Popier2017} the same author showed that the minimal solution to a certain singular BSDE yields a probabilistic representations of {\sl a} (possibly discontinuous) viscosity solution to the associated PDE. Our comparison result yields sufficient conditions for this minimal viscosity solution to be the unique (and hence continuous) solution. This complements the analysis is \cite{Ankirchner2014,Kruse2016}. {The existence (and uniqueness) of minimal solutions to BSDEs with singular terminal values for more general drivers has recently been established in \cite{Graewe2019} under (suitable regularity and) boundedness assumptions on the model parameters.} The framework in  \cite{Schied2013a} allows for unbounded coefficients but requires strong a priori estimates on the market impact term that are not satisfied in our main example. Complementing the analysis in  \cite{Schied2013a} our results show when value function derived in terms of Dawson-Watson superprocesses therein solves the HJB equation in the viscosity sense. 

The remainder of this paper is organized as follows. In Section \ref{main}, we summarize our main results. The existence of viscosity solution is proved in Section \ref{existence}; the verification argument is carried out in Section \ref{veri-viscosity}. Section \ref{non-markov} is devoted to an extension of our uniqueness result to a non-Markovian model with unbounded coefficients. 

\textit{Notation.}
We denote by $C_b(\mathbb R^d)$ the set of all functions $\phi:\mathbb R^d\rightarrow\mathbb R$ which are continuous and bounded on $\mathbb R^d$.  For a given $m\geq0,$ we define $C_m(\mathbb R^d)$ to be set of continuous functions that have at most polynomial growth of order $m$, i.e. the set of functions $\phi\in C(\mathbb R^d)$ such that $$\psi:=\frac{\phi(y)}{1+|y|^m} \in C_b(\mathbb R^d).$$ This space is a Banach space when endowed with the norm 
$$\|\phi\|_m:=\sup_{y\in\mathbb R^d}\frac{|\phi(y)|}{1+|y|^m}.$$
Let $I$ be a compact subset  of $\mathbb R$. A function $\phi$ belongs to $USC_m(I\times\mathbb R^d)$ (or $LSC_m(I\times\mathbb R^d))$ if it has at most polynomial growth of order $m$ in the second variable uniformly with respect to $t\in I$ and is upper (lower) semi-continuous on $I\times\mathbb R^d$.
Whenever the notation $T^-$ appears in the definition of a function space we mean the set of all functions whose restrictions satisfy the respective property when $T^-$ is replaced by any $s<T$, e.g., 
\[
	C_{m}([0,T^-]\times\mathbb R^d)=\{u:[0,T)\times\mathbb R^d\rightarrow \mathbb R: u_{|[0,s]\times\mathbb R^d}\in C_{m}([0,s]\times\mathbb R^d) \text{ for all } s\in[0,T)\}.
\]
Throughout, all equations and inequalities are to be understood in the a.s.\ sense. We adopt the convention that $C$ is a constant that may vary from line to line. 


\section{Assumptions and main results}\label{main}

For each initial state $(t,y,x)\in[0,T)\times\mathbb R^{d}\times\mathbb R$ we define by
\begin{equation} \label{value-function}
	V(t,y,x):=\inf_{(\xi,\mu)\in\mathcal A(t,x)}E\left[\int_t^T\eta(Y_s^{t,y}) |\xi_s|^2+\theta\gamma(Y_s^{t,y})|\mu_s|^2+\lambda(Y_s^{t,y})|X_s^{\xi,\mu}|^2\,ds\right]
\end{equation}
the {\sl value function} of the control problem (\ref{P1}) subject to the state dynamics 
\begin{equation}
\begin{split}
	dY^{t,y}_s &= b(Y^{t,y}_s) ds + \sigma(Y^{t,y}_s) dW_s, \quad Y^{t,y}_t=y \\
	dX_s^{\xi,\mu} &=- \xi_s\,ds- \mu_s\,dN_s, \quad X_t=x.
\end{split}
\end{equation} 
Here, $\xi=(\xi_s)_{s\in[t,T]}$ describes the {\sl rates} at which the agent trades in the primary market, while $\mu=(\mu_s)_{s\in[t,T]}$ describes the {\sl orders} submitted to the dark pool. The infimum is taken over the set $\mathcal A(t,x)$ of all {\sl admissible controls}, that is, {over all pairs of controls $(\xi,\mu)$ such that $\xi$ is progressively measurable, such that $\mu$ is predictable\footnote{We show later that we restrict ourselves to monotone portfolio processes so we could just as well assume that $\mu$ is bounded.}} and such that the resulting state process 
\begin{equation*} 
X_s^{\xi,{\mu}}=x-\int_t^s\xi_r\,dr-\int_t^s\mu_r\,dN_r,  \qquad t\leq s\leq T,
\end{equation*}
satisfies the terminal state constraint 
\begin{equation}   \label{liquidation-constraint}
	X_{T}^{\xi,\mu}=0.
\end{equation}
{The expected costs associated with an admissible liquidation strategy $(\xi,\mu)$ are given by 
$$J(t,y,x;\xi,\mu):=\mathbb E\left[\int^T_t c(Y^{t,y}_s, X^{\xi,\mu}_s,\xi_s,\mu_s) \,ds\right],$$
where the running cost function $c(y,x,\xi,\mu)$ is given by 
%
$$c(y,x,\xi,\mu):=\eta(y)|\xi|^2+\theta\gamma(y)|\mu|^2+\lambda(y)|x|^2.$$}
{
\begin{remark}
	We assume that the cost function is quadratic in the controls and the state variable. A generalization to general powers $p>1$ as in \cite{Graewe2018} can be established using similar arguments but renders the notation more cumbersome. 
\end{remark}}
The dynamic programming principle suggests that the value function satisfies the HJB equation
\begin{equation} \label{hjb}
-\partial_t V(t,y,x)-\mathcal L V(t,y,x)-\inf_{\xi,\mu\in\mathbb R}	H(t,y,x,\xi,\mu,V)=0, \quad (t,y,x)\in[0,T)\times\mathbb R^d\times\mathbb R,
\end{equation}
{where
\[ 
	\mathcal L:=\frac{1}{2}\tr(\sigma\sigma^* D_y^2)+\left\langle b,D_y\right\rangle 
\]
denotes the infinitesimal generator of the factor process }and the Hamiltonian $H$ is given by
\[
	H(t,y,x,\xi,\mu,V):=-\xi \partial_xV(t,y,x)+\theta(V(t,y,x-\mu)-V(t,y,x))+c(y,x,\xi,\mu).
\]
The quadratic cost function suggests an ansatz of the form $V(t,y,x)=v(t,y)|x|^2$.
The following result confirms this intuition. Its proof can be found in \cite[Section 2.2]{Graewe2018}.  

\begin{lemma} \label{lemma-hjb}
A nonnegative function $v:[0,T)\times\mathbb R^d\rightarrow[0,\infty)$ is a (sub/super) solution to the PDE
\begin{equation} \label{inflator-pde}
-\partial_tv(t,y)-\mathcal Lv(t,y)-F(y,v(t,y))=0,
\end{equation} 
where 
\begin{equation} \label{nonlinearity}
	F(y,v):= \lambda(y)-\frac{|v|^{2}}{\eta(y)}+\frac{\theta\gamma(y)v}{\gamma(y)+|v|}-\theta v,
\end{equation}
if and only if $v(t,y)|x|^2$ is a (sub/super) solution to the HJB equation~\eqref{hjb}.  In this case the infimum in~\eqref{hjb} is attained at 
\begin{equation} \label{optimal-control-hjb} 
\xi^*(t,y,x)=\frac{v(t,y)}{\eta(y)}x \quad \text{ and } \quad \mu^*(t,y,x)=\frac{v(t,y)}{\gamma(y)+v(t,y)}x
\end{equation}
and
\begin{equation} \label{optimal-hamiltonian}
H(t,y,x,\xi^*(t,y,x),\mu^*(t,y,x),v(\cdot,\cdot)|\cdot|^{2})=F(y,v(t,y))|x|^2.
\end{equation}
\end{lemma}


\subsection{Assumptions}

In order to prove the existence of a unique non-negative continuous viscosity solution of polynomial growth to our HJB equation we assume throughout that the factor process 
\begin{equation} \label{sde}
 Y_s^{t,y}=y+\int_t^sb(Y_r^{t,y})\,dr+\int_t^s\sigma(Y_r^{t,y})\,dW_r, \qquad t\leq s\leq T.
\end{equation}
satisfies the following condition. 

\begin{assumption}\label{Ass1}
The coefficients $b:\mathbb R^d\rightarrow \mathbb R^d$ and $\sigma:\mathbb R^d\rightarrow\mathbb R^{d\times \tilde d}$ are Lipschitz continuous.
\end{assumption}

The preceding assumption guarantees that the SDE \eqref{sde} has a unique strong solution $(Y_s^{t,y})_{s\in[t,T]}$ for every initial state $(t,y)\in[0,T]\times\mathbb R^d$ and that the mapping $(s,t,y) \mapsto Y^{t,y}_s$ is a.s.~continuous. We repeatedly use the following well known estimates; cf.~\cite[Corollary 2.5.12]{Krylov1980}. For all $m\geq 0,$ there {exists a constant $C>0$} such that for all $y\in\mathbb R^d, 0\leq t\leq s\leq T,$
\begin{equation}\label{Y-estimate}
\begin{aligned}
\mathbb E\sup_{t\leq s\leq T} |Y^{t,y}_s|^m&\leq C(1+|y|^m).\\
\end{aligned}
\end{equation}

Furthermore, we assume that the cost coefficients are continuous and of polynomial growth and that $\eta$ is twice continuously differentiable and satisfies a mild boundedness condition. 

\begin{assumption}\label{Ass2}
The cost coefficients satisfy the following conditions: 
\begin{itemize}
	\item[(i)] The coefficients $\eta,\gamma,\lambda,1/\eta:\mathbb R^d\rightarrow \mathbb [0,\infty)$ are continuous and of polynomial growth.  
	\item[(ii)] $\eta\in C^2$ and $\|\frac{\mathcal L \eta }{\eta}\|$ is bounded.
\end{itemize}
\end{assumption}
\begin{remark}
The preceding assumption is satisfied if, for instance $Y$ is {a} geometric Brownian motion or an Ornstein-Uhlenbeck (OU) process and
\[
	\eta(y) =1+|y|^2. 
\]
In both cases, condition (2.13) in \cite{Schied2013a} is violated. Our assumptions are also weaker than those in \cite{Graewe2018}. For instance, OU processes do not generate analytic semigroups, they do not satisfy the assumptions therein. 
\end{remark}



\subsection{Main results}
Before stating our first main result, we recall the notion of viscosity solutions for parabolic equations that will be used in this paper. The following definition can be found in \cite[Section 8]{Crandall1992}.
\begin{definition}
For semicontinuous functions $v:[0,T)\times\mathbb R^d\rightarrow \mathbb R$ we use the following solution concepts for the parabolic PDE:
\begin{equation} \label{general-pde}
-\partial_t v(t,y)-G(t,y,v(t,y),D_yv(t,y),D^2_yv(t,y))=0,
\end{equation}
where $G:[0,T)\times\mathbb R^d\times\mathbb R\times\mathbb R^d\times\mathbb S^d\rightarrow \mathbb R$ and $\mathbb S^d$ denotes the set of symmetric $d\times d$ matrices.
\begin{itemize}
	\item[(i)] $v \in USC_m([0,T^-]\times\mathbb R^d)$ is a \textit{(strict) viscosity subsolution} if for every $\varphi\in C^{1,2}_{loc}([0,T)\times \mathbb R^d)$ such that $\varphi\geq v$ and $\varphi(t,y)=v(t,y)$ at a point $(t,y)\in[0,T)\times\mathbb R^d$ it holds \[-\partial_t \varphi(t,y)-G(t,y,v(t,y),D_y\varphi(t,y),D^2_y\varphi(t,y))(<)\leq  0.\]
	\item[(ii)] $v\in LSC_m([0,T^-]\times\mathbb R^d)$ is a \textit{(strict) viscosity supersolution} if for every $\varphi\in C^{1,2}_{loc}([0,T)\times \mathbb R^d)$ such that $\varphi\leq v$ and $\varphi(t,y)=v(t,y)$ at a point $(t,y)\in[0,T)\times\mathbb R^d$ it holds \[-\partial_t \varphi(t,y)-G(t,y,v(t,y),D_y\varphi(t,y),D^2_y\varphi(t,y))(>)\geq  0.\]
	\item[(iii)] $v$ is a \textit{viscosity solution} if $v$ is both viscosity sub- and supersolution.
\end{itemize}
\end{definition}

We are now ready to state the main result of this paper. Its proof is given in Section 3 below. 

\begin{theorem} \label{thm-existence}  \leavevmode Under Assumptions~\ref{Ass1}, \ref{Ass2},  the singular terminal value problem
\begin{equation} \label{pde-v}
	\left\{\begin{aligned}	&{-\partial_t v}(t,y)-\mathcal L v(t,y)- F(y,v(t,y))=0,    & (t,y)\in[0,T)\times\mathbb R^d,&\\
&\lim_{t\rightarrow T}v(t,y)=+\infty  & \text{locally uniformly on } \mathbb R^d,&
\end{aligned}\right.
\end{equation}
with the nonlinearity $F$ given in~\eqref{nonlinearity} admits a unique nonnegative viscosity solution in \[
	C_{m}([0,T^-]\times\mathbb R^d)
\] 
for some $m\geq 0.$
\end{theorem}

The next result states that both the value function and the optimal controls are given in terms of the unique viscosity solution to the HJB equation.  The particular form of the feedback has been established in the literature before. What the proposition shows is that having a continuous viscosity solution to the HJB equation is enough to carry out the verification argument. 

\begin{proposition} \label{thm-verifcation}
Under Assumptions~\ref{Ass1} ,\ref{Ass2}, let $v$ be the unique nonnegative viscosity solution to the singular terminal value problem \eqref{pde-v}.  Then, the value function \eqref{value-function} is given by $V(t,y,x)=v(t,y)|x|^2$, and the optimal control $(\xi^*,\mu^*)$ is given in feedback form by
\begin{equation} \label{optimal-control}
		\xi_s^*= \frac{v(s,Y_s^{t,y})}{\eta(Y_s^{t,y})}X_s^* \quad \text{ and } \quad \mu_s^*=\frac{v(s,Y_s^{t,y})}{\gamma(Y_s^{t,y})+v(s,Y_s^{t,y})}X_{s-}^*.
\end{equation}
In particular, the resulting optimal portfolio process $(X^*_s)_{s\in[t,T]}$ is given by
\begin{equation} \label{optimal-position-process}	
X_s^*=x\exp\left(-\int_t^s\frac{v(r,Y_r^{t,y})}{\eta(Y_r^{t,y})}\,dr\right)\prod_{t<r\leq s}^{\Delta N_r\neq 0}\left(1-\frac{v(t,Y_r^{t,y})}{\gamma(Y_r^{t,y})+v(t,Y_r^{t,y})}\right).
\end{equation}
\end{proposition}



Let us close this section with a model of optimal portfolio liquidation where market impact is driven by an Ornstein-Uhlenbeck process while market risk is driven by a geometric Brownian motion. Specifically, let $Y=(Y^1, Y^2)$ be the diffusion process given by 
\[
	dY^1_t = - Y^1_t dt + dW^1_t \quad \mbox{and} \quad \frac{dY^2_t}{Y^2_t} = \sigma dW^2_t,
\]
where $W^1$ and $W^2$ are two (possibly correlated) Brownian motions, and let 
\[
	\eta(Y)=
	\left\{
	\begin{aligned}
	&1+|Y^1|^2,\quad  ~ \mbox{ if } Y^1<0,\\
	&\frac{1}{1+|Y^1|^2},\quad  \mbox{ if } Y^1\geq 0,
	\end{aligned}
	\right.
	\quad \gamma(Y) = 1, \quad \mbox{and} \quad 
	\lambda(Y) = \sigma^2 |Y^2|^2.
\]
The process $Y^1$ specifies a liquidity indicator that fluctuates around a stationary level (normalized to zero) with the market impact increasing when below average liquidity is available and decreasing when above average liquidity is available. Instantaneous market risk, on the other hand is captured by the volatility of the portfolio value assuming that asset prices follow a geometric Brownian motion. For the above choice of model parameters all assumptions on the cost and diffusion coefficients are satisfied. Hence, there exists a unique optimal liquidation strategy. 

\vspace{2mm}

\begin{remark}
To the best of our knowledge, numerical methods for simulating solutions to general PDEs with singular terminal values are still to be developed. At least two problems arise when simulating solutions to HJB equations with singular terminal state constraint. The most obvious problem is the singular terminal condition. This problem can potentially be overcome by noting that the function
\[
	w(t,y):=(T-t)v(t,y), \quad (t,y) \in [0,T) \times \mathbb{R}
\]
satisfies
the following PDE with finite terminal value, yet singular driver (see \cite{Graewe2018,Graewe2019} and Section 3 for details)
  \begin{equation*} 
	\left\{\begin{aligned}	&{-\partial_t w}(t,y)-\mathcal L w(t,y)-\frac{w(t,y)}{T-t}-(T-t)F(y,\frac{w(t,y)}{T-t})=0,    & (t,y)\in[0,T)\times\mathbb R^d,&\\
&\lim_{t\rightarrow T}w(t,y)=\eta(y) & \text{on } \mathbb R^d&
\end{aligned}\right. .
\end{equation*}
{The knowledge of a unique classical solution to the transformed problem opens up the possibility to apply higher-order numerical schemes and obtain accurate solutions in acceptable computing time. One possibility could be to study a one-to-one mapping of the unbounded control set to a compact set combined with a discretisation of the control, similar to the idea applied to an optimal investment problem in \cite{Reisinger2016}; an alternative approach based on monotonicity arguments is outlined in \cite{Graewe2019}.} The second problem is to fix appropriate boundary conditions (in space) for the numerical simulations; a similar problem arises if the binding state constraint is replaced by a finite penalty term. The analysis in Section 3 shows that for the benchmark case of a risk neutral investor $(\sigma = 0)$, 
\[
	w(t,y) \leq C \eta(y), \quad (t,y) \in [0,T) \times \mathbb{R}
\]       
for some $C>0$ from which we deduce zero boundary conditions if $\eta(y) \to 0$ for $|y| \to \infty$. In general we can not expect the above inequality to be an equality, though, not even asymptotically when $|y| \to \infty$. If we choose $\sigma = 0$ and the dynamics 
\[
	dY_t = -\tanh(Y_t - Y^3_t)dt + dW_t
\] 
for the liquidity index, then the index is mean-reverting to the levels $\pm 1$, the ``regimes of average liquidity''. Choosing $\eta(y) = \frac{1}{1+y^2}$ all our assumptions on the model parameters are satisfied. In this case we may regard the interval $(-1,+1)$ as the low and the set $[-1,1]^c$ as the high liquidity regime. Since $w(t,y) \to 0$ as $|y| \to \infty$, {the boundary problem can be dealt with.}
%
\end{remark}

\section{Solution and verification}
 
\subsection{Existence of solutions}\label{existence} 
In this section, we prove Theorem \ref{thm-existence}. In a first step, we establish a comparison principle for semicontinuous viscosity solutions to (\ref{pde-v}). In view of the singular terminal state constraint we can not follow the usual approach of showing that if a l.s.c. supersolution dominates an u.s.c. subsolution at the boundary, then it also dominates the subsolution on the entire domain. Instead, we prove that if some form of asymptotic dominance holds at the terminal time, then dominance holds near the terminal time. 

In a second step, we construct smooth sub- and supersolutions  to \eqref{pde-v} that satisfy the required asymptotic dominance condition. Subsequently, we apply Perron's method to establish an u.s.c. subsolution and a l.s.c. supersolution that are bounded from above/below by the smooth solutions. From this, we infer that the semi-continuous solutions can be applied to the comparison principle,  which then implies the existence of the desired continuous viscosity solution. 


\subsubsection{Comparison principle} 

Throughout this section, we fix $\delta\in(0,T]$ and for some $m\geq 0,$ let $\overline u \in LSC_m([T-\delta,T^-]\times\mathbb R^d)$ and $\underline u \in USC_m([T-\delta,T^-]\times\mathbb R^d)$ be a viscosity super- and a viscosity subsolution to \eqref{pde-v}. 

\begin{proposition} \label{comparison}
Under Assumptions~\ref{Ass1}, \ref{Ass2}, if, uniformly on $\mathbb R^d$,  
\begin{equation}\label{asympotic}
\limsup\limits_{t \rightarrow T}\frac{\underline u(t,y)(T-t)-\eta(y)}{1+|y|^{m}}\leq 0\leq\liminf\limits_{t\rightarrow T}\frac{\overline u(t,y)(T-t)-\eta(y)}{1+|y|^{m}},
\end{equation}
and 
\begin{equation}\label{interval}
\underline u(t,y)(T-t), \overline u(t,y)(T-t)\geq  \frac{1}{2}\eta(y),\quad t\in[T-\delta, T),
\end{equation}
then $$\underline u \leq \overline u\quad \text{on} \quad [T-\delta,T)\times\mathbb R^d.$$
\end{proposition}

Assumptions \eqref{asympotic}, \eqref{interval} are uncommon in the viscosity literature. However, we shall only use the comparison result to establish the existence of a solution, not the uniqueness. As a result, we only need to guarantee that the  semi-continuous solutions established through Perron's method satisfy both assumptions. 

The proof of the comparison principle is based on three auxiliary results. The first lemma is taken from \cite[Lemma A.2]{Graewe2018}. It is a modification of \cite[Lemma 3.7]{Barles1997}. 

\begin{lemma}
The difference $w:=\underline{u}-\overline{u}\in USC_m([T-\delta,T^-]\times\mathbb R^d)$ is a viscosity subsolution to
\begin{equation}\label{diff-pde} 
	-\partial_t w(t,y)-\mathcal L w(t,y)-l(t,y)w(t,y)=0,    \quad (t,y)\in [T-\delta,T)\times\mathbb R^d,
\end{equation}
where
$$l(t,y):=\frac{F(y,\underline{u}(t,y))-F(y,\overline{u}(t,y))}{\underline{u}(t,y)-\overline{u}(t,y)}\mathbb I_{\underline{u}(t,y)\neq\overline{u}(t,y)}.$$
\end{lemma}

The next lemma constructs a smooth strict supersolution to \eqref{diff-pde} of polynomial growth.

\begin{lemma}\label{bar}
For every $n\in\mathbb N$, there exists $K_n$ large enough such that
$$\chi(t,y):=\frac{e^{K_n(T-t)}(1+|y|^2)^{\frac{n}{2}}}{T-t}$$
satisfies
$$-\partial_t\chi(t,y)-\mathcal L \chi(t,y)+\frac{\chi(t,y)}{T-t}>0,    \quad (t,y)\in [T-\delta,T)\times\mathbb R^d.$$
\end{lemma}
\begin{proof}
Direct calculations verify that $h(t,y):=e^{K_n(T-t)}(1+|y|^2)^{\frac{n}{2}}$ satisfies $-\partial_t h(t,y)-\mathcal L h(t,y)>0$ in $[T-\delta,T)\times\mathbb R^d$ when $K_n$ is chosen sufficiently large; see also \cite[Proposition 5]{Alvarez1996}. Here it is used that $b$ and $\sigma$ are Lipschitz and thus are of linear growth. Hence, 
\begin{equation*}
\begin{aligned}
-\partial_t\chi(t,y)-\mathcal L \chi(t,y)+\frac{\chi(t,y)}{T-t}&=\frac{-\partial_t h(t,y)-\mathcal L h(t,y)}{T-t}>0.
\end{aligned}
\end{equation*}
\end{proof}

The following lemma is key to the proof of the comparison principle. 

\begin{lemma}\label{Phi}
If $n\in\mathbb N$ in Lemma \ref{bar}  is chosen large enough, then independent of $\alpha>0,$ the function 
$$\Phi_{\alpha}(t,y):= w(t,y)-\alpha\chi(t,y)$$
is either nonpositive or attains its supremum at some point $(t_{\alpha},y_{\alpha})$ in $[T-\delta,T)\times\mathbb R^d.$
\end{lemma}
\begin{proof}
Suppose that the supremum of $\Phi_{\alpha}$ on $[T-\delta,T)\times\mathbb R^d$ is positive and denote by $(t_k,y_k)$ a sequence in $[T-\delta,T)\times\mathbb R^d$ approaching the supremum point.  
The representation
$$\Phi_{\alpha}(t,y)=\frac{\left[\frac{\underline u(t,y)(T-t)-\eta(y)}{1+|y|^{m}}-\frac{\overline u(t,y)(T-t)-\eta(y)}{1+|y|^{m}}\right](1+|y|^{m})-\alpha e^{K_n(T-t)}(1+|y|^2)^{\frac{n}{2}}}{T-t},$$
along with condition \eqref{asympotic} shows that for any $n>m,$ 
$$\limsup\limits_{t\rightarrow T} \Phi_{\alpha}(t,y)=-\infty,  \textrm{ uniformly on }\mathbb R^d.$$
Hence $\lim\limits_{k} t_k<T.$ Furthermore,  {$w\in USC_m([T-\delta,T^-]\times\mathbb R^d)$} is bounded by a function of polynomial growth uniformly away from the terminal time. Choosing $n$ large enough this shows that $\lim\limits_{k} |y_k| < \infty$. 
As a result, the supremum is attained at some point $(t_{\alpha},y_{\alpha})$ because $\Phi_{\alpha}$ is upper semicontinuous. This proves the assertion. 
\end{proof}

We are now ready to prove the comparison principle. 

\begin{proof}[Proof of Proposition \ref{comparison}]
Let us fix $\alpha>0.$ By letting $\alpha\rightarrow 0$ it is sufficient to show that the function $\Phi_{\alpha}$ is nonpositive. 

In view of Lemma \ref{Phi}, we just need to consider the case where there exists a point $(t_{\alpha},y_{\alpha}) \in [T-\delta,T)\times\mathbb R^d$ such that 
%
$$w(t,y)-\alpha\chi(t,y)\leq w(t_{\alpha},y_{\alpha})-\alpha\chi(t_{\alpha},y_{\alpha}),\quad (t,y)\in [T-\delta,T)\times\mathbb R^d.$$
This inequality can be interpreted as $w-\psi_{\alpha}$ having a global maximum at $(t_{\alpha},y_{\alpha})$, where
$$\psi_{\alpha}:=\alpha\chi(t,y)+(w-\alpha\chi)(t_{\alpha},y_{\alpha}).$$
Since $\psi_{\alpha}$ is smooth and $w$ is a viscosity subsolution to \eqref{diff-pde}, 
$$-\partial_t \psi_{\alpha}(t_{\alpha},y_{\alpha})-\mathcal L \psi_{\alpha}(t_{\alpha},y_{\alpha})-l(t_{\alpha},y_{\alpha})w(t_{\alpha},y_{\alpha})\leq 0.$$
By the mean value theorem {along with the monotonicity of  $\partial_u F$, condition \eqref{interval} and the fact that $\partial_vF(y,v) \leq -\frac{2v}{\eta(y)}$ we get that} 
\begin{equation}\label{diff-bound}
l(t,y)=\frac{F(y,\underline{u}(t,y))-F(y,\overline{u}(t,y))}{\underline{u}(t,y)-\overline{u}(t,y)}\mathbb I_{\underline{u}(t,y)\neq\overline{u}(t,y)}\leq \partial_v F(y,\frac{\eta(y)}{2(T-t)}) {\color{red} \leq} -\frac{1}{T-t}.
\end{equation}
Thus, Lemma \ref{bar} implies 
\begin{equation}
\begin{aligned}
0 \geq &-\partial_t \psi_{\alpha}(t_{\alpha},y_{\alpha})-\mathcal L \psi_{\alpha}(t_{\alpha},y_{\alpha})-l(t_{\alpha},y_{\alpha})w(t_{\alpha},y_{\alpha}) \\
=&\alpha[-\partial_t \chi(t_{\alpha},y_{\alpha})-\mathcal L \chi(t_{\alpha},y_{\alpha})-l(t_{\alpha},y_{\alpha})w(t_{\alpha},y_{\alpha})]\\
>&-\alpha\frac{\chi(t_{\alpha},y_{\alpha})}{T-t_{\alpha}} -l(t_{\alpha},y_{\alpha})w(t_{\alpha},y_{\alpha})\\
\geq & \alpha l(t_{\alpha},y_{\alpha})\chi(t_{\alpha},y_{\alpha}) -l(t_{\alpha},y_{\alpha})w(t_{\alpha},y_{\alpha})\\
=&-l(t_{\alpha},y_{\alpha})\Phi_{\alpha}(t_{\alpha},y_{\alpha}).
\end{aligned}
\end{equation}
Since $l\leq 0,$  we can conclude that $\Phi_{\alpha}(t_{\alpha},y_{\alpha})\leq 0,$ thus $\Phi_{\alpha}\leq 0.$
\end{proof}


\subsubsection{Existence via Perron's method}\label{perron}

Armed with our comparison principle, the existence of a viscosity solution to our HJB equation can be established using Perron's method as soon as suitable sub- and supersolutions can be identified.  In view of Assumption \ref{Ass2},  $\eta, \lambda\in C_m(\mathbb R^d)$ for some $m\geq 0$ and $\|\frac{\mathcal L \eta }{\eta}\|$ is well-defined and finite. Hence 
\begin{equation} \label{delta}
	\delta :=1/\|\frac{\mathcal L \eta }{\eta}\| \wedge T> 0.\footnote{{We use the convention $1/0=\infty.$}} 
\end{equation}	
By a direct computation, we can find a constant $K'$ large enough such that the function: $\hat h(t,y):=e^{K'(T-t)}(1+|y|^2)^{m/2}$ satisfying 
\[
	 -\partial_t \hat h(t,y)-\mathcal{L}\hat h(t,y)-\lambda(y)\geq 0.
\]
Let us then define
\begin{equation*}
	 \check v(t,y):=\frac{\eta(y)-\eta(y)\|\frac{\mathcal L \eta }{\eta}\|(T-t)}{e^{\theta(T-t)}(T-t)} \quad \mbox{and} \quad 
	 \hat v(t,y):=\frac{\eta(y)+\eta(y)\|\frac{\mathcal L \eta }{\eta}\|(T-t)}{(T-t)}+\hat h(t,y).
\end{equation*}

\begin{proposition}\label{local-estimate}
Under Assumption \ref{Ass1}, \ref{Ass2} the functions $\check v, \hat v$ are a nonnegative classical sub- and supersolution to~\eqref{pde-v} on $[T-\delta,T)\times\mathbb R^d,$ respectively.
\end{proposition}
\begin{proof}
To verify the supersolution property of $\hat v$, we first verify that
\begin{equation}\label{derivative-v-hat}
\begin{split} 
	& -\partial_t\hat v(t,y)-\mathcal L \hat v(t,y) \\
	= &-\frac{\eta(y)+\mathcal L \eta(y)(T-t)+\mathcal L \eta(y)\|\frac{\mathcal L \eta }{\eta}\|(T-t)^2}{(T-t)^2}-\partial_t \hat h(t,y)-\mathcal{L}\hat h(t,y)
\end{split}
\end{equation}
Recalling the definition~\eqref{nonlinearity} of $F$, we have since $\hat v\geq0$,
\[
	-F(y,\hat v(t,y))\geq -\lambda(y)+\frac{\hat v(t,y)^{2}}{\eta(y)}.
\]
Next, we apply the inequality $(u+v+w)^2\geq u^2+2uv$ for $u,v,w\geq 0$ to the term $\hat v(t,y)^2$ to obtain
\begin{equation} \label{nonlinearity-v-hat}
	-F(y,\hat v(t,y))\geq-\lambda (y)+\frac{\eta(y)^2+2\eta(y)^2\|\frac{\mathcal L \eta }{\eta}\|(T-t)}{\eta(y) (T-t)^2}.
\end{equation}
Adding \eqref{derivative-v-hat} and \eqref{nonlinearity-v-hat} yields
\begin{equation*}
\begin{aligned} 
	-\partial_t\hat v(t,y)-\mathcal L \hat v(t,y)-F(y,\hat v(t,y))\geq&\frac{2\eta(y)\|\frac{\mathcal L \eta }{\eta}\|-\mathcal L \eta(y)-\mathcal L \eta(y)\|\frac{\mathcal L \eta }{\eta}\|(T-t)}{(T-t)}\\
	&-\partial_t \hat h(t,y)-\mathcal{L}\hat h(t,y)-\lambda(y).
\end{aligned}
\end{equation*}
The definition of $\delta$ yields $1\geq \|\frac{\mathcal L \eta }{\eta}\|(T-t)$ for $t\in [T-\delta, T)$ and so, 
\begin{equation*}
\begin{aligned} 
	&2\eta(y)\|\frac{\mathcal L \eta }{\eta}\|-\mathcal L \eta(y)-\mathcal L \eta(y)\|\frac{\mathcal L \eta }{\eta}\|(T-t)\\
	\geq& \eta(y)\|\frac{\mathcal L \eta }{\eta}\|\cdot\left[1+\|\frac{\mathcal L \eta }{\eta}\|(T-t)\right]-\mathcal L \eta(y)-\mathcal L \eta(y)\|\frac{\mathcal L \eta }{\eta}\|(T-t)\\
	=&\left[1+\|\frac{\mathcal L \eta }{\eta}\|(T-t)\right]\cdot\left[\eta(y)\|\frac{\mathcal L \eta }{\eta}\|-\mathcal L \eta(y)\right]\geq 0.
\end{aligned}
\end{equation*}
We conclude that
\begin{equation*}
	-\partial_t\hat v(t,y)-\mathcal L \hat v(t,y)-F(y,\hat v(t,y))\geq 0.
\end{equation*}
Next, we verify the subsolution property of $\check v$. By direct computation,
\begin{equation} \label{derivative-v-check}
	-\partial_t\check v(t,y)-\mathcal L \check v(t,y)=-\frac{\eta(y)+\mathcal L \eta(y)(T-t)-\mathcal L \eta(y)\|\frac{\mathcal L \eta }{\eta}\|(T-t)^2}{e^{\theta(T-t)}(T-t)^2}-\theta \check v(t,y).
\end{equation}
On the other hand, since $\lambda,\gamma\geq 0$, and $\check v\geq 0$ on $[T-\delta, T)\times\mathbb R^d$,
\[
	-F(y,\check v(t,y))\leq \frac{\check v(t,y)^2}{\eta(y)}+\theta \check v(t,y).
\]
We estimate $\check v(t,y)^2$ using the inequality $(u-v)^2\leq u^2-uv$ for $u\geq v\geq 0$ and obtain,
\begin{equation} \label{nonlinearity-v-check}
-F(y,\check v(t,y))\leq\frac{\eta(y)-\eta(y)\|\frac{\mathcal L \eta }{\eta}\|(T-t)}{e^{2\theta(T-t)}(T-t)^2}+\theta \check v(t,y).
\end{equation}
{Since $e^{-2\theta(T-t)}\leq e^{-\theta(T-t)}$}, adding \eqref{derivative-v-check} and \eqref{nonlinearity-v-check} yields
\begin{equation*}
-\partial_t\check v(t,y)-\mathcal L \check v(t,y)-F(t,\check v(t,y))\leq -\frac{\eta(y)\|\frac{\mathcal L \eta }{\eta}\|+\mathcal L \eta(y)-\mathcal L \eta(y)\|\frac{\mathcal L \eta }{\eta}\|(T-t)}{e^{\theta(T-t)}(T-t)}.
\end{equation*}
Using again that $1\geq \|\frac{\mathcal L \eta }{\eta}\|(T-t)$ we obtain, 
\begin{equation*}
\begin{aligned} 
	&\eta(y)\|\frac{\mathcal L \eta }{\eta}\|+\mathcal L \eta(y)-\mathcal L \eta(y)\|\frac{\mathcal L \eta }{\eta}\|(T-t)\\
	\geq & \eta(y)\|\frac{\mathcal L \eta }{\eta}\|\cdot\left[1-\|\frac{\mathcal L \eta }{\eta}\|(T-t)\right]+\mathcal L \eta(y)-\mathcal L \eta(y)\|\frac{\mathcal L \eta }{\eta}\|(T-t)\\
	=&\left[1-\|\frac{\mathcal L \eta }{\eta}\|(T-t)\right]\cdot\left[\eta(y)\|\frac{\mathcal L \eta }{\eta}\|+\mathcal L \eta(y)\right]\geq 0.
\end{aligned}
\end{equation*}
Thus,
\begin{equation*}
-\partial_t\check v(t,y)-\mathcal L \check v(t,y)-F(t,\check v(t,y))\leq 0.
\end{equation*}
\end{proof}
\begin{proof}[\textbf{Proof of Theorem \ref{thm-existence}.}]
From the definition of $\check v, \hat v$ we have
\begin{equation} \label{order}
\begin{aligned}
	(T-t)\check v(t,y)&=\eta(y)+ \eta(y) O(T-t) \quad \text{uniformly in $y$ as $t\rightarrow T$.}\\
   (T-t)\hat v(t,y)&=\eta(y)+ (1+|y|^{m})O(T-t) \quad \text{uniformly in $y$ as $t\rightarrow T$.}
\end{aligned}
\end{equation}
Then for $\varepsilon=\frac{1}{2},$ there exists $\delta_0\in(0,\delta]$ such that for all $t\in [T-\delta_0,T),$
\begin{equation*}
\check v(t,y)(T-t)>\eta(y)-\frac{1}{2} \eta(y)=\frac{1}{2}\eta(y)\quad\text{ uniformly on } \mathbb R^d.
\end{equation*} 
Since $\eta\in C_m(\mathbb R^d),$ we obtain from \eqref{order} that
\begin{equation} \label{limit-v}
\begin{aligned}
\lim\limits_{t \rightarrow T}\frac{\check v(t,y)(T-t)-\eta(y)}{1+|y|^{m}}=\lim\limits_{t \rightarrow T}\frac{\hat v(t,y)(T-t)-\eta(y)}{1+|y|^{m}}=0,\quad\text{ uniformly on } \mathbb R^d.
\end{aligned}
\end{equation}  

In order to apply Perron's method, we set
 $$\mathcal S=\{u | u\text{ is a subsolution of }\eqref{pde-v} \text{ on } [T-\delta_0,T)\times\mathbb R^d\text{ and }u\leq \hat v\}.$$ From Proposition \ref{local-estimate} we know that $\check v\in \mathcal S,$ so $\mathcal S$ is non-empty. Thus, the function
\[
 v(t,y)=\sup \{u(t,y):u\in\mathcal S\}
\]
is well-defined. Classical arguments\footnote{ 
The standard Perron method of finding viscosity solutions for elliptic PDEs can be found in \cite{Crandall1992}. We refer to \cite [Appendix A] {Zhan1999} for the proof of this method for parabolic equations.} show that the upper semi-continuous envelope $v^*$ is a viscosity subsolution to~\eqref{pde-v}. From \cite [Lemma~A.2] {Zhan1999}, the lower semi-continuous envelope $v_*$ of $v$ is also a viscosity supersolution to~\eqref{pde-v}. Since $\check v\leq v_*\leq v^*\leq \hat v, $ we have that for all $t\in [T-\delta_0,T),$
\[
v_*(t,y)(T-t), v^*(t,y)(T-t)\geq \frac{1}{2}\eta(y), \quad\text{ uniformly on } \mathbb R^d.
\]
and
\begin{align*}
	\frac{\check v(t,y)(T-t)-\eta(y)}{1+|y|^{m}}\leq \frac{v_*(t,y)(T-t)-\eta(y)}{1+|y|^{m}} \leq & \frac{v^*(t,y)(T-t)-\eta(y)}{1+|y|^{m}} \\ \leq & \frac{\hat v(t,y)(T-t)-\eta(y)}{1+|y|^{m}}.
\end{align*}
Hence, it follows from \eqref{limit-v} that, 
\begin{equation} 
\begin{aligned}
\lim\limits_{t \rightarrow T}\frac{v_*(t,y)(T-t)-\eta(y)}{1+|y|^{m}}=\lim\limits_{t \rightarrow T}\frac{v^*(t,y)(T-t)-\eta(y)}{1+|y|^{m}}=0,\quad\text{ uniformly on } \mathbb R^d.
\end{aligned}
\end{equation}  
From our comparison principle [Proposition \ref{comparison}] we conclude that $v^*\leq v_*\text{ on } [T-\delta_0,T)\times\mathbb R^d,$ which shows that $v$ is the desired viscosity solution to \eqref{inflator-pde} that belongs to $C_m([T-\delta_0,T^-]\times\mathbb R^d)$. 

By\cite [Remark 6]{Alvarez1996}, there exists a unique viscosity solution $v\in C_m([0,T-\delta_0]\times\mathbb R^d)$ to \eqref{inflator-pde} when imposed at $t=T-\delta_0$ with a terminal value in $C_m(\mathbb R^d).$ Hence from the comparison principle for {\sl continuous} viscosity solutions \cite [Lemma 3.1]{Graewe2018}, we get a unique global viscosity solution $$v\in C_m([0,T^-]\times\mathbb R^d).$$
\end{proof}

\begin{remark}
 If all the coefficients of the generator $F$ and the SDE \eqref{sde} are bounded, then one can show that twice differentiability of $\eta$ is not needed; only a uniform continuity is required to choose continuous solutions which satisfying the conditions \eqref{asympotic} and \eqref{interval}. Thus a unique viscosity solution can be obtained by the same argument above.
\end{remark}


\subsection{Verification}\label{veri-viscosity}

	This section is devoted to the verification argument. Throughout,  $v\in C_m([0,T^-]\times\mathbb R^d)$ denotes the unique nonnegative viscosity solution to the singular terminal value problem \eqref{pde-v}. We will prove that the viscosity solution is indeed the value function to our stochastic control problem.

In a first step we are now going to show that the feedback control given in \eqref{optimal-control} is indeed admissible. 
\begin{lemma}	\label{lemma-admissible}
	The pair of feedback controls $(\xi^*,\mu^*)$ given by \eqref{optimal-control} is admissible. 
\end{lemma}
\begin{proof}
	{Given the feedback form in \eqref{optimal-control}, one can easily obtain that the pair of controls $(\xi^*,\mu^*)$ is admissible and the resulting portfolio process $(X_s^*)_{s\in[t,T]}$ is monotone. It remains to verify the liquidation constraint.} 
Since $\check{v}\leq v\leq \hat v$ on $[T-\delta,T)$ where $\delta$ is defined in (\ref{delta}), it holds for any $r\in [T-\delta,T)$ that,
\begin{equation*}\label{bound}
\frac{1-\|\frac{\mathcal L\eta}{\eta}\|(T-r)}{e^{\theta(T-r)}(T-r)}\eta(Y_r^{t,y})\leq v(r,Y_r^{t,y})\leq \frac{1+\|\frac{\mathcal L\eta}{\eta}\|(T-r)}{T-r}\eta(Y_r^{t,y})+\hat h(r,Y_r^{t,y}).
\end{equation*}
For $s\in [T-\delta,T),$
\begin{equation} \label{liquidation}
	\begin{aligned}
		|X_s^*| &\leq |x|\exp\left(-\int_{t}^s\frac{v(r,Y_r^{t,y})}{\eta(Y_r^{t,y})}\,dr\right)\\
	&\leq|x|\exp\left(-\int_{T-\delta}^{s}\frac{v(r,Y_r^{t,y})}{\eta(Y_r^{t,y})}\,dr\right) \\
	&\leq|x| \exp\left(-\int_{T-\delta}^{s} \frac{1-\|\frac{\mathcal L\eta}{\eta}\|(T-r)}{e^{\theta(T-r)}(T-r)}\,dr\right) \\
	&\leq|x| \exp\left(\int_{T-\delta}^{s} \frac{e^{\theta(T-r)}-[1-\|\frac{\mathcal L\eta}{\eta}\|(T-r)]}{e^{\theta(T-r)}(T-r)}\,dr\right) \exp\left(-\int_{T-\delta}^{s} \frac{1}{T-r}\,dr\right) \\
	&\leq|x| \exp\left(\int_{T-\delta}^{s}\left[ \frac{e^{\theta(T-r)}-1}{e^{\theta(T-r)}(T-r)}+\frac{\|\frac{\mathcal L\eta}{\eta}\|}{e^{\theta(T-r)}}\right]\,dr\right)\cdot \frac{T-s}{\delta}\\
	&\leq C|x| \frac{T-s}{\delta}.
	\end{aligned}
\end{equation}
The last inequality holds because $\lim\limits_{r\rightarrow T} \frac{e^{\theta(T-r)}-1}{e^{\theta(T-r)}(T-r)}=\theta.$ As a result, $X^*_{T-}=0$ and hence  $X^*_T=0.$

%
\end{proof}


It has been shown in \cite[Lemma 5.2]{Graewe2018} that we may w.l.o.g restrict ourselves to admissible controls that result in a monotone portfolio process. We denote by $\bar{\mathcal A}(t, x)$ the set of all admissible controls under which the portfolio process is monotone. 

Next, we give a probabilistic representation of the viscosity solution to \eqref{pde-v}. In \cite{Popier2017}, the author showed that the possibly discontinuous minimal solution of a certain backward stochastic differential equation with singular terminal condition gives a probabilistic representation of the minimal viscosity solution of an associated partial differential equation; continuity of the solution was {\sl not} established. However, continuity is necessary to carry out the verification argument. We obtain a solution to the corresponding FBSDE in a different way since the existence of the (continuous) viscosity solution has already been proved. 
{
\begin{proposition}\label{prop-BSDE}
 Under Assumptions \ref{Ass1}, \ref{Ass2}, for any fixed $\epsilon\in(0,T)$ and $(t,y)\in [\epsilon,T)\times\mathbb R^d$, there exists a pair of processes $ (U^{t,y}, Z^{t,y})\in S^2_{\mathcal{F}}(t,T;\mathbb R)\times L^2_{\mathcal{F}}(t,T;\mathbb R^{1 \times \tilde d})$ satisfying that $U^{t,y}_t=v(t-\epsilon,y)$ and for any $\epsilon\leq t\leq r\leq s\leq T,$
\[
U^{t,y}_r=U^{t,y}_s+\int^s_r F(Y^{t,y}_{\rho},U^{t,y}_{\rho})d\rho- \int^s_r Z^{t,y}_{\rho} dW_{\rho}.
\]

\end{proposition}
\begin{proof}
We consider the forward-backward system  
\begin{equation}\label{BSDE}
\left\{
\begin{aligned}
dY_s&=b(Y_s)ds+\sigma(Y_s)dW_s, &s\in[t,T],&\\
dU_s&=-f(s,Y_s)ds+Z_sdW_s,& s\in[t,T],&\\
Y_t&=y, U_{T}=v(T-\epsilon,Y_{T}),& &
\end{aligned}\right.
\end{equation}
and the corresponding PDE
\begin{equation}\label{originalpde}
\left\{
\begin{aligned}
&-w_t(t,y)-\mathcal Lw(t,y)-f(t,y)=0, &(t,y)\in [\epsilon,T)\times R^d,&\\
&w(T, y)=v(T-\epsilon,y), & y\in \mathbb R^d&
\end{aligned}\right.
\end{equation}
where $f(t,y):=F(y,v(t-\epsilon,y))$ and $F$ is defined in (\ref{nonlinearity}). Recalling the polynomial growth condition on the cost coefficients in Assumption \ref{Ass2} and the polynomial growth property of the solution $v$ established in Theorem \ref{thm-existence}, we know that $f\in C_{m^\prime}([\epsilon,T]\times\mathbb R^d)$, for some $m^{\prime}\geq m.$ Together with Assumption \ref{Ass1} and the fact that $ v(T-\epsilon,\cdot)\in C_{m}(\mathbb R^d)$, we conclude from \cite[Theorem 2.1]{Karoui1997} that the system admits a unique solution 
\[
	(Y^{t,y}, U^{t,y}, Z^{t,y}) \in S^2_{\mathcal{F}}(t,T;\mathbb R^d)\times S^2_{\mathcal{F}}(t,T;\mathbb R)\times L^2_{\mathcal{F}}(t,T;\mathbb R^{1 \times \tilde d}). 
\] 

Let $ w(t,y):=U^{t,y}_t.$ By the Feynman-Kac formula \cite[Theorem 3.2]{Pardoux1999}, $w$ is the unique viscosity solution of \eqref{originalpde} with driver $f$. Due to the time-homogeneity of the PDE in \eqref{pde-v}, viscosity solutions stay viscosity solutions when shifted in time. Let $\tilde v(t,y):=v(t-\epsilon,y)$ on $[\epsilon,T]$. By the definition of $f,$ we see that $\tilde v$ is also a viscosity solution of \eqref{originalpde} with driver $f$ on $[\epsilon,T]$. Hence it follows that $w=\tilde v$. By the Markov property, we have for any $r\in[t,T]$ that $0\leq U^{t,y}_r=v(r-\epsilon,Y^{t,y}_r).$ Thus $U^{t,y}$ is also a solution to the following FBSDE:
\begin{equation*}
\left\{
\begin{aligned}
dY_s&=b(Y_s)ds+\sigma(Y_s)dW_s, &s\in[t,T],&\\
dU_s&=-F(Y_s, U_s)ds+Z_sdW_s, &s\in[t,T],&\\
Y_t&=y, U_{T}=v(T-\epsilon,Y_{T}).& &
\end{aligned}\right.
\end{equation*}

\end{proof}
For any $\epsilon\in(0,T),$ we can restrict our interval on $[t,T-\epsilon]$ and repeat the arguments above without shifting in time. This yields a solution $ (\tilde{U}^{t,y}, \tilde{Z}^{t,y})\in S^2_{\mathcal{F}}(t,T-\epsilon;\mathbb R)\times L^2_{\mathcal{F}}(t,T-\epsilon;\mathbb R^{1 \times \tilde d})$ satisfying that $\tilde{U}^{t,y}_t=v(t,y)$ and for any $0\leq t\leq r\leq s<T-\epsilon,$ 
\begin{equation*}
\left\{
\begin{aligned}
dY_s&=b(Y_s)ds+\sigma(Y_s)dW_s, &s\in[t,T-\epsilon],&\\
d\tilde{U}_s&=-F(Y_s, \tilde{U}_s)ds+\tilde{Z}_sdW_s, &s\in[t,T-\epsilon],&\\
Y_t&=y, \tilde{U}_{T-\epsilon}=v(T-\epsilon,Y_{T-\epsilon}).& &
\end{aligned}\right.
\end{equation*}
Since $\epsilon$ is arbitrary, a global solution on $[0,T)$ can be obtained.
\begin{corollary}
	Under Assumptions \ref{Ass1}, \ref{Ass2}, there exists  processes $ (\tilde{U}^{t,y}, \tilde{Z}^{t,y})\in S^2_{\mathcal{F}}(t,T^-;\mathbb R)\times L^2_{\mathcal{F}}(t,T^-;\mathbb R^{1 \times \tilde d})$ satisfying that $\tilde{U}^{t,y}_t=v(t,y)$ and for any $0\leq t\leq r\leq s<T,$
	\begin{equation}\label{coro-BSDE}
\tilde{U}^{t,y}_r=\tilde{U}^{t,y}_s+\int^s_r F(Y^{t,y}_{\rho},\tilde{U}^{t,y}_{\rho})d\rho- \int^s_r \tilde{Z}^{t,y}_{\rho} dW_{\rho}.
	\end{equation}

\end{corollary}

The following lemma is key to the verification argument.

\begin{lemma} \label{ito} Fix $\epsilon\in(0,T)$ and $(t,y)\in [\epsilon,T)\times\mathbb R^d$. For every $(\xi,\mu)\in\bar{\mathcal A}(t,x)$ and $s\in[t,T)$,
\begin{equation*} 
	v(t-\epsilon,y)|x|^2 \leq \mathbb E\left[v(s-\epsilon,Y_s^{t,y})|X_s^{\xi,\mu}|^2\right]+ \mathbb E\left[\int_t^sc(Y_r^{t,y},X_r^{\xi,\mu},\xi_r,\mu_r)\,dr \right].
\end{equation*}
\end{lemma}
\begin{proof}
By Proposition \ref{prop-BSDE}, we know that $(U^{t,y},Z^{t,y})$ solves the following BSDE:
$$U^{t,y}_t=U_{s}^{t,y}+\int^s_tF(Y^{t,y}_r,U^{t,y}_r)dr-\int^s_t Z^{t,y}_rdW_r.$$
This allows us to apply to $U_{s}^{t,y}|X_{s}^{\xi,\mu}|^2$ the classical integration by parts formula for semimartingales in order to obtain
\begin{multline*}
	U_{t}^{t,y}|x|^2 = U_{s}^{t,y}|X_{s}^{\xi,\mu}|^2 +\int_t^{s}\big\{ F(Y^{t,y}_r,U^{t,y}_r)|X_r^{\xi,\mu}|^2\\
	+2\xi_r U^{t,y}_r\sgn(X_r^{\xi,\mu})|X_r^{\xi,\mu}|-\theta U^{t,y}_r(|X_r^{\xi,\mu}-\mu_r|^2-|X_r^{\xi,\mu}|^2)\big\}\,dr\\ 
	-\int_t^{s}\sigma(Y_r^{t,y}) Z^{t,y}_r|X_r^{\xi,\mu}|^2\,dW_r  -\int_t^{s} U^{t,y}_r(|X_{r-}^{\xi,\mu}-\mu_r|^2-|X_{r-}^{\xi,\mu}|^2)\,d\widetilde N_r,
\end{multline*}
where $\widetilde N_r= N_r-\theta r$ denotes the compensated Poisson process. Moreover,  $|X^{\xi,\mu}|\leq |x|$ and $|\mu|\leq |x|$, due to the monotonicity of the portfolio process. Furthermore,$$\int_t^{s}\sigma(Y_r^{t,y}) Z^{t,y}_r|X_r^{\xi,\mu}|^2\,dW_r$$
is a uniformly integrable martingale because
$$2\mathbb E\left[\left(\int_t^{s}|\sigma(Y_r^{t,y})|^2\cdot |Z^{t,y}_r|^2|X_r^{\xi,\mu}|^4\,dr\right)^{1/2}\right]\leq \mathbb E\left(\sup_{t\leq r\leq s}|\sigma(Y_r^{t,y})|^2+|x|^4\int^s_t|Z^{t,y}_r|^2\,dr\right)<\infty.$$
As a consequence, the above stochastic integrals are true martingales. Hence, recalling \eqref{optimal-hamiltonian},
\begin{align}
	U_{t}^{t,y}|x|^2 
	&=\mathbb E\left[U_{s}^{t,y}|X_{s}^{\xi,\mu}|^2\right] +\mathbb E\left[\int_t^{s} c(Y_r^{t,y},X_r^{\xi,\mu},\xi_r,\mu_r)\,dr \right] \nonumber\\
	&\quad+\mathbb E\left[\int_t^{s} \big\{ F(Y^{t,y}_r,U^{t,y}_r)|X_r^{\xi,\mu}|^2 -H(r,Y_r^{t,y},X_r^{\xi,\mu},\xi_r,\mu_r,U_{r}^{t,y}|X_r^{\xi,\mu}|^2)\big\}\,dr \right]\nonumber\\
	&\leq \mathbb E\left[U_{s}^{t,y}|X_{s}^{\xi,\mu}|^2\right]  +\mathbb E\left[\int_t^{s} c(Y_r^{t,y},X_r^{\xi,\mu},\xi_r,\mu_r)\,dr \right]. \label{suboptimal-estimate}	
\end{align}
Since $U_{t}^{t,y}=v(t-\epsilon,y), U^{t,y}_r=v(r-\epsilon,Y^{t,y}_r),$ we have 
\[
	v(t-\epsilon,y)|x|^2 \leq 
	\mathbb E\left[v(s-\epsilon,Y_{s}^{t,y})|X_{s}^{\xi,\mu}|^2\right] + 
	\mathbb E\left[\int_t^{s} c(Y_r^{t,y},X_r^{\xi,\mu},\xi_r,\mu_r)\,dr \right].
\]
\end{proof}
We are now ready to carry out the verification argument.

\begin{proof}[Proof of Proposition \ref{thm-verifcation}]
Let $(\xi,\mu)\in\bar{\mathcal A}(t,x)$. 
By the liquidation constraint of $X^{\xi,\mu},$ letting $s \to T$ yields
\[\mathbb E\left[v(s-\epsilon,Y_{s}^{t,y})|X_{s}^{\xi,\mu}|^2\right] \rightarrow 0.\]
Hence,
\begin{equation*} 
 v(t-\epsilon,y)|x|^2\leq J(t,y,x;\xi,\mu).
\end{equation*}
Finally, by letting $\epsilon\rightarrow 0$, we conclude
\begin{equation*} 
v(t,y)|x|^2\leq J(t,y,x;\xi,\mu).
\end{equation*}
 on $[0,T)\times\mathbb R^d$ by the continuity of $v$ and the nonnegativity of $J$.

Using similar arguments to the proof of Proposition \ref{ito} on the BSDE \eqref{coro-BSDE}, we can obtain that
\[ v(t,y)|x|^2 \leq\mathbb E\left[v(s,Y_s^{t,y})|X_s^{\xi,\mu}|^2\right]+\mathbb  E\left[\int_t^sc(Y_r^{t,y},X_r^{\xi,\mu},\xi_r,\mu_r)\,dr \right].\]
By Lemma~\ref{lemma-hjb} equality holds in the preceding inequality if $\xi=\xi^{*}$ and $\mu=\mu^{*}.$
Thus,
\begin{equation*}
\begin{aligned}
v(t,y)|x|^2 &= \mathbb E\left[v(s,Y_s^{t,y})|X_s^{\xi^*,\mu^*}|^2\right]+\mathbb  E\left[\int_t^sc(Y_r^{t,y},X_r^{\xi^*,\mu^*},\xi^*_r,\mu^*_r)\,dr
\right]\\
&\geq \mathbb  E\left[\int_t^sc(Y_r^{t,y},X_r^{\xi^*,\mu^*},\xi^*_r,\mu^*_r)\,dr
\right]
\end{aligned}
\end{equation*}
from which we conclude that
\begin{equation*} 
v(t,y)|x|^2\geq J(t,y,x;\xi^{*},\mu^{*}).
\end{equation*}

This shows that the strategy $(\xi^*,\mu^*)$ is indeed optimal. \end{proof}
}


\section{Uniqueness in the non-Markovian framework} \label{non-markov}

{In this section we assume that the filtration is solely generated by the Brownian motion.} The existence of a minimal nonnegative solution 
\[
(\mathcal Y,\mathcal Z)\in L_{\mathcal F}^2(\Omega;C([0,T^-];\mathbb R_+))\times L_\mathcal F^2(0,T^-;\mathbb R^{1\times \tilde{d}})
\] 
to the BSDE
\begin{equation} \label{BSDE-Y}
-d\mathcal{Y}_t=\left\{\lambda_t-\frac{|\mathcal{Y}_t|^{2}}{\eta_t}\right\}dt-\mathcal{Z}_t\,dW_t, \quad 0\leq t<T; \quad \lim_{t\rightarrow T} \mathcal{Y}_t=+\infty
\end{equation}
has been established in~\cite{Ankirchner2014} under the assumption that $\eta\in L^2_\mathcal F(0,T;\mathbb R_+)$, $\eta^{-1}\in L^1_{\mathcal F}(0,T;\mathbb R_+)$, $\lambda \in L^2_{\mathcal F}(0,T^-;\mathbb R_+)$, and $\mathbb E[\int_0^T(T-t)^2\lambda_t\,dt]<\infty.$ 

In this section we extend our uniqueness result to non-Markovian models and prove the existence of a unique nonnegative solution under the following conditions; they correspond to those in the Markovian setting.
\begin{assumption}\label{Ass3}
	\begin{itemize} 
		\item[(i)] The process $\eta$ is a positive It\^o diffusion satisfying that $d\eta_t=\alpha_t\,dt+\beta_t\,dW_t$ with $(\alpha,\beta) \in L^2_{\mathcal F}(0,T;\mathbb R\times\mathbb R^{1\times \tilde{d}})$.
		\item[(ii)]The processes $\eta,\eta^{-1}\in L^{2}_{\mathcal F}(\Omega;C([0,T];\mathbb R))$ and $\eta^{-1}\alpha\in L^\infty_{\mathcal F}(0,T;\mathbb R)$.
		\item[(iii)] There exists  a positive It\^o diffusion $h_t$ such that $d h_t=\alpha^{\prime}_t\,dt+\beta^{\prime}_t\,dW_t$ with $(\alpha^{\prime},\beta^{\prime}) \in L^2_{\mathcal F}(0,T;\mathbb R\times\mathbb R^{1\times \tilde{d}})$ and $h^{-1}\lambda, h^{-1}\alpha^{\prime} \in L^\infty_{\mathcal F}(0,T;\mathbb R)$. 
	\end{itemize} 
\end{assumption}
\begin{proposition} \label{a-priori-bound}
	Let Asssumption \ref{Ass3} hold. Set $\tau:=1/\|\eta^{-1}\alpha\|_{L^\infty}\wedge T$ and $\tilde{K}:=\| h^{-1}\alpha^{\prime}\|_{L^\infty}+\|h^{-1}\lambda\|_{L^\infty}.$ For any solution $$(\mathcal{Y}, \mathcal{Z})\in L_{\mathcal F}^2(\Omega;C([0,T^-];\mathbb R_+))\times L_\mathcal F^2(0,T^-;\mathbb R^{1\times \tilde{d}})$$ to~\eqref{BSDE-Y} the following estimates hold for $T-\tau\leq t< T$:
	\begin{equation} \label{bsde-estimate}
	\eta_t\left(\frac{1}{T-t}-\|\eta^{-1}\alpha\|_{L^\infty}\right)\leq \mathcal{Y}_t\leq \eta_t\left(\frac{1}{T-t}+\|\eta^{-1}\alpha\|_{L^\infty}\right)+e^{\tilde K(T-t)}h_t. 
	\end{equation}
\end{proposition}
\begin{proof} For $0<\epsilon<\tau$ we define $(\overline{\mathcal{Y}_t}^\epsilon)_{t\in[T-\tau,T-\epsilon)}$ by
	\[\overline{\mathcal{Y}_t}^\epsilon=\eta_t\left(\frac{1}{T-\epsilon-t}+\|\eta^{-1}\alpha\|_{L^\infty}\right)+e^{\tilde K(T-\epsilon-t)}h_t.\]
	We will show that these processes are supersolutions to \eqref{BSDE-Y} but with the singularity at $t=T-\epsilon$, \[\lim_{t\rightarrow T-\epsilon}\overline{\mathcal{Y}_t}^\epsilon=+\infty.\] 
	Precisely, 
	\[-d\overline{\mathcal{Y}_t}^\epsilon=g^\epsilon(t,\overline{\mathcal{Y}_t}^\epsilon)\,dt-\overline{\mathcal{Z}_t}^\epsilon\,dW_t, \qquad T-\tau\leq t<T-\epsilon,\]
	where 
	\begin{align*}
	g^\epsilon(t,\overline{\mathcal{Y}_t}^\epsilon):=&-\frac{\eta_t}{(T-\epsilon-t)^2}-\alpha_t\left(\frac{1}{T-\epsilon-t}+\|\eta^{-1}\alpha\|_{L^\infty}\right)\\
	&+\tilde{K}e^{\tilde K(T-\epsilon-t)}h_t-e^{\tilde K(T-\epsilon-t)}\alpha^{\prime}_t
	\end{align*}
	and $\overline{\mathcal{Z}}^{\epsilon}\in\bigcap_{t\in[T-\tau,T-\epsilon)} L^2_{\mathcal F}(T-\tau,t;\mathbb R^{1\times\tilde{d}})$.
	A calculation as in the proof of Proposition \ref{local-estimate} verifies that for all $T-\tau\leq t<T-\epsilon$,
	\[g^\epsilon(t,\overline{\mathcal{Y}_t}^\epsilon)\geq \lambda_t-\frac{|\overline{\mathcal{Y}_t}^\epsilon|^2}{\eta_t}=:f(t,\overline{\mathcal{Y}_t}^\epsilon).\]
	 {Indeed, applying the inequality $(u+v+w)^2\geq u^2+2uv$ for $u,v,w\geq 0$ to the term $|\overline{\mathcal{Y}_t}^\epsilon|^2$,  we obtain that
\[\frac{|\overline{\mathcal{Y}_t}^\epsilon|^2}{\eta_t}\geq \frac{\eta_t^2+2\eta_t^2\|\eta^{-1}\alpha\|_{L^\infty}(T-\epsilon-t)}{\eta_t(T-\epsilon-t)^2}.
\]
Noting that $\tau:=1/\|\eta^{-1}\alpha\|_{L^\infty}\wedge T$, $\|\eta^{-1}\alpha\|_{L^\infty}(T-\epsilon-t)\leq 1$ for $t\in [T-\tau, T-\epsilon),$ we have that
\begin{equation*}
\begin{aligned}
&\frac{|\overline{\mathcal{Y}_t}^\epsilon|^2}{\eta_t}-\frac{\eta_t}{(T-\epsilon-t)^2}-\alpha_t\left(\frac{1}{T-\epsilon-t}+\|\eta^{-1}\alpha\|_{L^\infty}\right)\\
\geq &\frac{2\eta_t\|\eta^{-1}\alpha\|_{L^\infty}-\alpha_t(1+\|\eta^{-1}\alpha\|_{L^\infty}(T-\epsilon-t))}{(T-\epsilon-t)}\\
\geq &\frac{\eta_t\|\eta^{-1}\alpha\|_{L^\infty}(1+\|\eta^{-1}\alpha\|_{L^\infty}(T-\epsilon-t))-\alpha_t(1+\|\eta^{-1}\alpha\|_{L^\infty}(T-\epsilon-t))}{(T-\epsilon-t)}\\
=&\frac{(\eta_t\|\eta^{-1}\alpha\|_{L^\infty}-\alpha_t)(1+\|\eta^{-1}\alpha\|_{L^\infty}(T-\epsilon-t))}{(T-\epsilon-t)}\\
\geq &0
\end{aligned}
\end{equation*}
Recalling that $\tilde{K}:=\| h^{-1}\alpha^{\prime}\|_{L^\infty}+\|h^{-1}\lambda\|_{L^\infty}$, we have that
$\tilde{K}e^{\tilde K(T-\epsilon-t)}h_t-e^{\tilde K(T-\epsilon-t)}\alpha^{\prime}_t\geq e^{\tilde K(T-\epsilon-t)}\lambda_t\geq \lambda_t.$
Therefore, we can conclude that $g^\epsilon(t,\overline{\mathcal{Y}_t}^\epsilon)\geq \lambda_t-\frac{|\overline{\mathcal{Y}_t}^\epsilon|^2}{\eta_t}.$}
	
	We now consider the difference of $\mathcal{Y}$ and $\overline{\mathcal{Y}}^\epsilon$ for $T-\tau\leq t\leq s<T-\epsilon$:
	\begin{align*}
	\overline{\mathcal{Y}_t}^\epsilon-\mathcal{Y}_t&=\mathbb E\left[\overline{\mathcal{Y}_s}^\epsilon-\mathcal{Y}_s+\int_t^s g^\epsilon(r,\overline{\mathcal{Y}_r}^\epsilon)\,dr-\int_t^s f(r,\mathcal{Y}_r)\,dr\Big|\mathcal{F}_t\right]\\
	& \geq\mathbb E\left[\overline{\mathcal{Y}_s}^\epsilon-\mathcal{Y}_s+\int_t^s f(r,\overline{\mathcal{Y}_r}^\epsilon)-f(r,\mathcal{Y}_r)\,dr\Big|\mathcal{F}_t\right]\\
	&=\mathbb E\left[\overline{\mathcal{Y}_s}^\epsilon-\mathcal{Y}_s+\int_t^s (\overline{\mathcal{Y}_r}^\epsilon-\mathcal{Y}_r)\Delta_r\,dr\Big|\mathcal{F}_t\right]
	\end{align*}
	where 
	\begin{equation*}
	{\Delta_r}=\left\{
	\begin{aligned}
	&\frac{f(r,\overline{\mathcal{Y}_r}^\epsilon)-f(r,\mathcal{Y}_r)}{\overline{\mathcal{Y}_r}^\epsilon-\mathcal{Y}_r},&\quad \textit{ if } \quad \overline{\mathcal{Y}_r}^\epsilon-\mathcal{Y}_r\neq 0,&&\\
	&0,&\quad else.&&
	\end{aligned}\right.
	\end{equation*}
	Note that $\Delta\leq 0.$ By the explicit representation of the solution to linear BSDEs, 
	\begin{equation*}
	\overline{\mathcal{Y}_t}^\epsilon-\mathcal{Y}_t\geq\mathbb E\left[(\overline{\mathcal{Y}_s}^\epsilon-\sup_{t\leq s\leq T-\epsilon}\mathcal{Y}_s)\exp\left(\int_t^s\Delta_r\,dr\right)\right].
	\end{equation*}
	{Since $\overline{\mathcal{Y}_s}^\epsilon\geq 0,$ $\mathbb E[\sup_{t\leq s\leq T-\epsilon}\mathcal{Y}_s]<+\infty$ due to $\mathcal Y\in  L_{\mathcal F}^2(\Omega;C([0,T^-];\mathbb R_+)),$ we can apply Fatou's lemma to the expectation above as $s\rightarrow T-\epsilon$ to obtain that $\overline{\mathcal{Y}_t}^\epsilon-\mathcal{Y}_t\geq 0.$} Taking $\epsilon\rightarrow 0$ we obtain the upper estimate. The lower estimate can be established by similar arguments.
\end{proof}
{
\begin{lemma}
	Suppose that Asssumption \ref{Ass3} holds. Let $(\mathcal{Y}, \mathcal{Z})$ be a solution of \eqref{BSDE-Y} in the space $L_{\mathcal F}^2(\Omega;C([0,T^-];\mathbb R_+))\times L_\mathcal F^2(0,T^-;\mathbb R^{1\times \tilde{d}})$. Let $X^*_t=\exp(-\int^t_0\frac{\mathcal Y_s}{\eta_s}\,ds)$ denote the associated portfolio process. Then $X^*\mathcal Z\in L^2_{\mathcal F}(0,T;\mathbb R).$
\end{lemma}
\begin{proof}
	Let $M_t=\mathcal Y_tX^*_t+\int^t_0\lambda_sX^*_s\,ds.$ Integration by parts yields 
	\begin{equation}\label{M-t}
	dM_t=X^*_t\mathcal Z_tdW_t.
	\end{equation}
	Hence, $M$ is a nonnegative local martingale on $[0,T)$ and in particular a nonnegative supermartingale. Thus, it converges almost surely in $\mathbb R$ as $t$ goes to $T$.
	Similarly to \eqref{liquidation}, we use the lower estimate in \eqref{bsde-estimate} to obtain that for $s\in[T-\tau,T)$
	$$|X^*_s|\leq C(T-s).$$
 In view of the upper estimate in \eqref{bsde-estimate}, we have that
	$$\mathbb E\left[\sup_{T-\tau\leq t\leq s}|\mathcal Y_tX^*_t|^2\right]\leq C\mathbb E\left[\sup_{T-\tau\leq t\leq T}(|\eta_t|^2+|h_t|^2)\right],$$
	where the constant $C$ is independent of $s.$ Thus, applying the dominated convergence theorem implies
	\begin{equation*}
	\begin{aligned}
\mathbb E\left[\sup_{0\leq t\leq T}|M_t|^2\right]&\leq C\left(\mathbb E\left[\sup_{0\leq t\leq T-\tau}|\mathcal Y_t|^2\right]+\mathbb E\left[\sup_{T-\tau\leq t\leq T}(|\eta_t|^2+|h_t|^2)\right]+\mathbb E\left[\int^T_0|\lambda_s|^2\,ds\right]\right)\\
&<+\infty.
	\end{aligned}
	\end{equation*}
Recalling the equation \eqref{M-t}, we have that $X^*\mathcal Z\in L^2_{\mathcal F}[0,T;\mathbb R)$ and that $M$ is indeed a nonnegative martingale on $[0,T].$
	\end{proof}
It follows from \cite[{Proposition 4.4}]{Ankirchner2014} that $\mathcal Y$ is the minimal solution of \eqref{BSDE-Y}. Therefore, we can obtain the uniqueness result.
\begin{theorem}
	Under Asssumption \ref{Ass3}, there exists a unique solution to the BSDE~\eqref{BSDE-Y} in $L_{\mathcal F}^2(\Omega;C([0,T^-];\mathbb R_+))\times L_\mathcal F^2(0,T^-;\mathbb R^{1\times \tilde{d}})$.
\end{theorem}
}

\section{Conclusion}

In this paper we established a novel comparison principle for viscosity solutions to HJB equations with singular terminal conditions arising in models of optimal portfolio liquidation under market impact. Our method is flexible enough to allow for possibly unbounded coefficients. The comparison principle allowed us to prove the existence of a unique continuous viscosity solution to the HJB equation and hence the existence of a unique optimal trading strategy. Without continuity it is typically impossible to study further regularity properties of the value function. Using our continuity result it is possible to prove that the value function is a $\pi$-strong solution to the HJB equation under mild additional conditions on the model parameters. This means that the value function can be approximated by $C^{1,2}$ function uniformly on compact sets. Several other avenues are open for future research. For instance, it would clearly be desirable to weaken the regularity assumption on the unbounded market impact coefficient $\eta$. The regularity assumption was needed to carry out the Taylor-type approximation of the value function at the terminal time. 


\bibliographystyle{siam}
{\small
\bibliography{HX_bib}
}

\end{document}